\ifdefined\XeTeXversion\else\pdfoutput=1\fi
\documentclass[a4paper,onecolumn,11pt,unpublished,nopdfoutputerror]{quantumarticle}

\usepackage[utf8]{inputenc}
\usepackage[T1]{fontenc}
\usepackage[english]{babel}
\usepackage{amsmath,amssymb,amsthm,mathtools}
\usepackage{bm}
\usepackage{booktabs}
\usepackage{array}
\usepackage{graphicx}
\usepackage{tikz}
\usetikzlibrary{arrows.meta,calc,positioning}
\usepackage{microtype}
\usepackage{placeins}
\usepackage{float}
\usepackage[numbers,sort&compress]{natbib}
\usepackage[colorlinks=true,citecolor=quantumviolet,linkcolor=quantumviolet,urlcolor=quantumviolet]{hyperref}
\newtheorem{theorem}{Theorem}

\newtheorem{proposition}[theorem]{Proposition}
\newtheorem{corollary}[theorem]{Corollary}
\theoremstyle{definition}
\newtheorem{definition}{Definition}
\theoremstyle{remark}

\newcommand{\ket}[1]{\lvert #1\rangle}
\newcommand{\bra}[1]{\langle #1\rvert}

\newcommand{\expect}[1]{\left\langle #1\right\rangle}
\newcommand{\norm}[1]{\left\lVert #1\right\rVert}
\newcommand{\id}{\mathbb I}
\newcommand{\CCZ}{\operatorname{CCZ}}
\newcommand{\CZ}{\operatorname{CZ}}
\newcommand{\Hccz}{\ket{H_3}}
\newcommand{\Bell}{\mathcal B_\star}
\newcommand{\field}{\mathbb Q(\sqrt2)}

\newcommand{\rank}{\operatorname{rank}}

\begin{document}

\title{Device-Independent Self-Testing of the Three-Qubit CCZ Hypergraph State}

\author{Yunguang Han}
\affiliation{College of Computer Science and Technology, Nanjing University of Aeronautics and Astronautics, Nanjing 211106, P.~R.~China}
\author{Xin Kuang}
\affiliation{College of Computer Science and Technology, Nanjing University of Aeronautics and Astronautics, Nanjing 211106, P.~R.~China}
\author{Xingyuan Bu}
\affiliation{College of Computer Science and Technology, Nanjing University of Aeronautics and Astronautics, Nanjing 211106, P.~R.~China}
\author{Yukun Wang}
\email{wykun06@gmail.com}
\affiliation{Beijing Key Laboratory of Petroleum Data Mining, China University of Petroleum, Beijing 102249, People's Republic of China}
\affiliation{Key Lab of Processors, Institute of Computing Technology, CAS, Beijing 100190, People's Republic of China}

\maketitle

\begin{abstract}
The three-qubit CCZ state is the smallest rank-three hypergraph state and an
elementary entangled magic resource.  Its cubic phase is governed by
generalized stabilizers that are not Pauli strings, so standard graph-state
self-testing arguments do not apply directly.  We show that twenty
correlators, all obtainable from five of the eight global input triples in the
tripartite two-input, two-output scenario, determine this state and the action
of the Pauli $X/Z$ measurements up to local isometries.  The proof fixes eight
equally weighted computational branches and propagates conditional $X$-flip
relations across the branch cube, recovering the minus sign of the $111$
amplitude.  These five-context correlations are nonlocal,
but the canonical Pauli measurements cannot attain the largest quantum value
of any Bell inequality that they violate: whenever they maximize a Bell
expression, its local bound has the same value.  Introducing an independent
third measurement makes self-testing from maximal Bell violation possible.
We construct an explicit Bell inequality whose maximal quantum violation
self-tests the CCZ state and all three local measurements.  An exact
sum-of-squares decomposition proves the quantum bound, and its equality
conditions yield an analytic SWAP extraction.  Together, these results give
two explicit device-independent self-tests of the CCZ state and demonstrate
that determining a state and its measurements from several correlator
equalities is distinct from identifying them through the maximal violation of
a single Bell inequality.
\end{abstract}

\section{Introduction}
\label{sec:introduction}

A Bell experiment provides only conditional input--output probabilities; the
state, Hilbert spaces, and measurement operators remain uncharacterized.
Self-testing asks whether these probabilities nevertheless determine a
quantum realization up to local isometries and an auxiliary state that
factors from the extracted target system
\cite{MayersYao2004,SupicBowles2020}.  This requirement is stronger than Bell
nonlocality.  A violation rules out local hidden-variable models, whereas a
self-test identifies both the state and the action of selected measurements.

The operator information needed for self-testing is encoded in equality
conditions.  Saturated correlators yield state-dependent vector relations,
which can then be assembled into a local extraction map.  Sum-of-squares
decompositions expose the same relations through a Bell operator; NPA and
SWAP relaxations provide complementary numerical tools away from an ideal
point \cite{BampsPironio2015SOS,NPA2007,NPA2008,Bancal2015Swap,
Kaniewski2016Analytic}.  The three-qubit $W$ state has been robustly
self-tested using a dedicated Bell inequality \cite{Wu2014W}.  A
complementary device-independent tomography construction designed
permutationally invariant Bell inequalities for the $W$ state and used SWAP
relaxations to lower-bound its fidelity as a function of the observed
violation \cite{Pal2014Tomography}.  Other multipartite constructions use
only marginal information or cover broader families of symmetric
three-qubit states \cite{Li2018Marginal,Li2020SymmetricThreeQubit}.

For graph states, Pauli stabilizers already supply the local algebra required
by the extraction proof.  This leads to self-tests with data scaling linearly
in the graph and to scalable Bell inequalities with robust equality
conditions \cite{McKague2011GraphStates,Baccari2020GraphBell}.  Hypergraph
states replace graph edges by controlled phases acting on hyperedges of
arbitrary rank \cite{KruszynskaKraus2009,Qu2013Encoding,Rossi2013Hypergraph}.
Their natural generalized stabilizers contain controlled operations on the
remaining vertices of an incident hyperedge.  Except for special symmetric
subclasses with additional Pauli symmetries, these relations do not reduce to
Pauli strings \cite{Guhne2014Hypergraph,Lyons2015LUSymmetries,
Lyons2017LocalPauli}.  Recovering a higher-order phase from black-box
correlations therefore calls for a different mechanism from the usual
graph-state stabilizer argument.

We focus on the three-qubit state associated with a single rank-three
hyperedge,
\begin{equation}
\Hccz=\CCZ\ket{+}^{\otimes3}
=\frac1{\sqrt8}\sum_{a,b,c\in\{0,1\}}(-1)^{abc}\ket{abc}.
\label{eq:ccz-state-intro}
\end{equation}
The phase polynomial $abc$ changes only the computational branch $111$.
Correspondingly, the generalized stabilizers are
\begin{equation}
K_A=X_A\CZ_{BC},\qquad
K_B=X_B\CZ_{AC},\qquad
K_C=X_C\CZ_{AB}.
\label{eq:ccz-generalized-stabilizers}
\end{equation}
The phase is cubic in the computational bits, and the corresponding
stabilizer relation is nonlinear in the local $Z$ observables.  The state
$\Hccz$ is the minimal qubit instance containing a rank-three hyperedge.  The
entanglement and local-unitary structure of three-qubit and more general
hypergraph states have been studied in
Refs.~\cite{Guhne2014Hypergraph,Lyons2015LUSymmetries,
Qu2013ThreeQubit,Ghio2018HypergraphWitness}.  The simplicity of
Eq.~\eqref{eq:ccz-state-intro}
therefore hides the feature relevant here: one Boolean monomial changes a
single computational amplitude, while each local bit flip acquires a sign
conditioned on the other two bits.

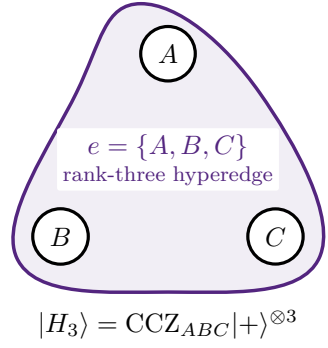
\begin{figure}[t]
\centering
\begin{tikzpicture}[
  scale=1.05,
  vertex/.style={circle,draw=black,very thick,fill=white,minimum size=7mm},
  hyperedge/.style={draw=quantumviolet,line width=1.3pt,fill=quantumviolet!8},
  every node/.style={font=\small}
]
\coordinate (A) at (0,1.45);
\coordinate (B) at (-1.35,-0.85);
\coordinate (C) at (1.35,-0.85);
\draw[hyperedge] plot[smooth cycle,tension=0.72] coordinates {
 (-0.34,2.08) (1.15,0.96) (1.90,-1.12)
 (0,-1.55) (-1.90,-1.12) (-1.15,0.96)
};
\node[vertex] at (A) {$A$};
\node[vertex] at (B) {$B$};
\node[vertex] at (C) {$C$};
\node[quantumviolet,fill=white,rounded corners=1pt,inner sep=2.2pt,
      align=center] at (0,0.12)
  {$e=\{A,B,C\}$\\[-0.7mm]
   \scriptsize rank-three hyperedge};
\node[align=center] at (0,-1.96)
  {$\Hccz=\CCZ_{ABC}\ket{+}^{\otimes3}$};
\end{tikzpicture}
\caption{Hypergraph representation of $\Hccz$.  The shaded region denotes the
single rank-three hyperedge $e=\{A,B,C\}$; no pairwise graph edges are
present.  The gate $\CCZ_{ABC}$ changes the sign of only the $111$
computational amplitude and gives generalized stabilizers such as
$K_A=X_A\CZ_{BC}$, together with cyclic permutations.}
\label{fig:ccz-hypergraph}
\end{figure}

The same rank-three phase has several operational roles.  Hypergraph resource
states generated by CCZ gates enable measurement-based constructions using
only Pauli measurements, realize nontrivial symmetry-protected resource
phases, and can change or parallelize the circuit-depth structure of a
computation \cite{MillerMiyake2016UnionJack,MillerMiyake2018SPTO,
MillerSandersMiyake2017,GachechiladzeGuhneMiyake2019,
Takeuchi2019HypergraphUniversal}.  The state $\CCZ\ket{+}^{\otimes3}$ is also
an entangled nonstabilizer resource for fault-tolerant computation.  CCZ and
Toffoli gates are related by a local Hadamard conjugation on the target
qubit.  Magic-state distillation supplies non-Clifford resources in
fault-tolerant architectures \cite{BravyiKitaev2005}, with dedicated
protocols for Toffoli and CCZ resource states developed in
Refs.~\cite{Jones2013Toffoli,Eastin2013Toffoli,HaahHastings2018CCZ}.
Resource-theoretic magic and the nonstabilizerness of
hypergraph states give complementary ways to quantify this role
\cite{HowardCampbell2017Magic,ChenYanZhou2024Magic}.  Experimentally,
four-qubit hypergraph
states and their Pauli-measurement processing have also been demonstrated on
a programmable photonic platform \cite{HuangEtAl2024Experiment}.

These applications have motivated a substantial literature on state
verification with trusted local measurements.  Stabilizer testing for
measurement-only computation was introduced for graph-state resources
\cite{HayashiMorimae2015}, extended to hypergraph-state verification
\cite{MorimaeTakeuchiHayashi2017}, and later generalized to all
polynomial-time-generated hypergraph states without assuming independent and
identically distributed preparations \cite{TakeuchiMorimae2018}.  General
optimality questions for verification with local measurements were studied in
Ref.~\cite{Pallister2018LocalVerification}.  For hypergraph states, Zhu and
Hayashi developed Pauli-measurement protocols whose efficiency is controlled
by the coloring and degree of the hypergraph
\cite{ZhuHayashi2019HypergraphVerification}, together with extensions to
adversarial preparation \cite{ZhuHayashi2019AdversarialPRL}.  Such protocols
test whether a supplied state is close to a specified target under a trusted
measurement model.  They do not infer the measurement operators from
input--output data.  Computational self-testing provides another model in
which cryptographic assumptions are used to certify a single quantum device
\cite{MetgerVidick2021Computational}.  In that model, Mizutani \emph{et al.}
showed that a CCZ magic state can be self-tested
\cite{Mizutani2022MagicSelfTesting}.  Their protocol and ours therefore
target the same nonstabilizer resource under different assumptions: the
former uses one device and cryptographic hardness, whereas the standard Bell
setting considered here uses three noncommunicating parties and neither
trusted measurements nor computational assumptions.

Bell inequalities tailored to hypergraph states have established
multipartite nonlocality, Hardy-type
contradictions, violations that grow with system size, and robustness to
particle loss; later work quantified robust Bell nonlocality for symmetric
families \cite{Guhne2014Hypergraph,Gachechiladze2016Hypergraph,
Noller2023SymmetricHypergraph}.  These results characterize nonclassicality.
Self-testing additionally requires a characterization of every strategy
attaining the largest quantum value and a local isometry that extracts the
target state and measurements.  The missing question is whether the cubic
CCZ phase itself can be recovered from black-box Bell correlations, and
whether one maximally violated Bell inequality can enforce the required
operator relations.

General existence theorems cover every pure bipartite entangled state
\cite{Coladangelo2017Bipartite} and every pure entangled multiqubit state up
to complex conjugation.  The universal multipartite construction in
Ref.~\cite{BalanzoJuando2026AllPure} uses at
most $9\cdot2^{n-2}-4$ binary measurements per party, which gives an upper
bound of fourteen for three-qubit states.  This bound reflects a protocol
designed for arbitrary target states rather than one optimized for a
structured state such as CCZ.  Recent constructions address scalable
certification of generic multipartite states with ancillary Bell pairs
\cite{Liu2026Scalable} and generate self-testing Bell inequalities from
tensor-network connectors \cite{Abiuso2025Renormalization}.  Those results
establish broad coverage.  By specializing to CCZ, the first construction
below uses only two binary measurements per party and exposes which
correlations recover its cubic phase and which relations enter the extraction
proof.  Graph-theoretic and custom
sum-of-squares methods provide systematic tools for such tailored Bell
inequalities \cite{Bharti2022GraphTheoretic,Barizien2024CustomBell}.

To the best of our knowledge, the results below give the first explicit
device-independent self-test tailored to the three-qubit CCZ hypergraph
state.  They also separate two certification tasks that coincide in many
familiar examples.  A collection of correlator equalities may determine a
state and its measurements, while maximal-violation self-testing requires a
single Bell expression with a strictly nonlocal largest quantum value.  This
distinction is related to the convex geometry of quantum correlations
\cite{GohEtAl2018Geometry,ChenEtAl2023Nonexposed}.

First, the analytic self-test uses the tripartite binary-output Bell scenario
with local input cardinalities $(2,2,2)$, taking $X/Z$ as the two inputs at
each site.  Its twenty correlators come from only five of the eight possible
global contexts.
Vanishing $Z$ moments make the eight computational
branches equiprobable.  The remaining equalities give generalized
stabilizers and a norm-square identity that fixes the three negative flips
incident on $111$.  Cross-party commutation propagates these signs over the
branch cube and closes the SWAP proof.  The argument assumes neither local
qubit dimension nor one-dimensional branch spaces.  The same five contexts
also give a Bell value $11$ above the local bound $9$, so the tested
correlations are strictly nonlocal.

Second, strict nonlocality is insufficient to make the canonical
measurements a maximally violating strategy.  If the CCZ state with Pauli
$X/Z$ measurements attains the largest quantum value $q$ of any Bell
expression, we prove that its local bound is also $q$.  Hence no Bell
inequality maximized by this canonical realization is violated by it.  This
conclusion concerns the canonical realization; rotated or otherwise
inequivalent two-setting measurements remain outside its scope.

Third, an independent reflection $D_p$ for each party makes
maximal-violation self-testing possible.  We construct a nine-orbit Bell
inequality with local and quantum bounds $258-36\sqrt2$ and
$42+120\sqrt2$.  Its maximal violation self-tests the state and all three
local measurements.  The upper bound follows from a finite noncommutative
sum-of-squares certificate over $\mathbb Q(\sqrt2)$; the equality conditions
lead to an analytic extraction proof.  Section~\ref{sec:quantitative}
records the nonideal norm estimates that follow directly from the two
proofs.

\section{Bell scenario and definition of self-testing}
\label{sec:scenario}

We use the standard tripartite tensor-product Bell model.  The parties are
$p\in\{A,B,C\}$, the shared state acts on
$\mathcal H_A\otimes\mathcal H_B\otimes\mathcal H_C$, and no finite-dimensional
bound is imposed.  A local binary measurement can be taken to be a Hermitian
reflection after a local Naimark dilation.  A mixed shared state can be
purified by assigning a purifying register to one party.  We may therefore
write the unknown realization as a normalized pure state $\ket\psi$ and local
reflections satisfying
\begin{equation}
(M_p^x)^2=\id,\qquad (M_p^x)^\dagger=M_p^x,
\qquad [M_p^x,M_q^y]=0\quad(p\ne q).
\label{eq:black-box-model}
\end{equation}
No commutation or anticommutation relation is assumed between different
measurements of the same party.

For a nonempty party subset $S$ and one selected measurement at each party in
$S$, a correlator is
\begin{equation}
E(M_{p_1}^{x_1}\cdots M_{p_k}^{x_k})
=\bra\psi M_{p_1}^{x_1}\cdots M_{p_k}^{x_k}\ket\psi.
\label{eq:correlator-definition}
\end{equation}
For binary outcomes, the correlators belonging to a fixed context determine
its full conditional probability distribution by Fourier inversion.  The
two-setting results use local inputs $Z_p,X_p$.  The Bell result uses a third
reflection $D_p$.  In the latter scenario $D_p$ is an independent black-box
input; the algebra never assumes a linear relation among $D_p$, $X_p$, and
$Z_p$.

The target state is Eq.~\eqref{eq:ccz-state-intro}.  Its reference
measurements are
\begin{equation}
Z_p^{\rm ref}=\sigma_z^{(p)},\qquad
X_p^{\rm ref}=\sigma_x^{(p)},
\label{eq:reference-xz}
\end{equation}
and, only in the reference realization of the three-setting protocol,
\begin{equation}
D_p^{\rm ref}=\frac{\sigma_x^{(p)}+\sigma_z^{(p)}}{\sqrt2}.
\label{eq:reference-d}
\end{equation}
The generalized stabilizer in Eq.~\eqref{eq:ccz-generalized-stabilizers} can
be written entirely in terms of the reference $Z$ and $X$ measurements.  For
example,
\begin{equation}
\CZ_{BC}=\frac{\id+Z_B+Z_C-Z_BZ_C}{2},
\label{eq:cz-polynomial}
\end{equation}
so that the relation $K_A\Hccz=\Hccz$ is visible through ordinary
correlators.

\begin{definition}[State-and-measurement self-testing]
Let $\mathbf c$ be correlation data generated by a reference state
$\ket{\psi_{\rm ref}}$ and reference measurements $M_{p,\rm ref}^x$.  We say
that these data self-test the state and measurements if every implementation
producing $\mathbf c$ admits a local isometry
$\Phi=\Phi_A\otimes\Phi_B\otimes\Phi_C$ and a normalized auxiliary state
$\ket{\mathrm{junk}}$ such that
\begin{align}
\Phi\ket\psi
 &=\ket{\psi_{\rm ref}}\otimes\ket{\mathrm{junk}},
\label{eq:self-test-state-definition}\\
\Phi M_p^x\ket\psi
 &=M_{p,\rm ref}^x\ket{\psi_{\rm ref}}
   \otimes\ket{\mathrm{junk}}
\label{eq:self-test-measurement-definition}
\end{align}
for the measurements specified by the theorem.
\end{definition}

The results below use a canonical local SWAP isometry.  Define the spectral
projectors of $Z_p$ by
\begin{equation}
P_a^p=\frac{\id+(-1)^aZ_p}{2},\qquad a\in\{0,1\},
\label{eq:branch-projectors}
\end{equation}
and set
\begin{equation}
\Phi_p\ket\varphi
=\ket0\otimes P_0^p\ket\varphi
+\ket1\otimes X_pP_1^p\ket\varphi.
\label{eq:canonical-isometry}
\end{equation}
The order $X_pP_1^p$ is fixed.  Equation~\eqref{eq:black-box-model} alone
gives
\begin{equation}
(P_0^p)^\dagger P_0^p
+(X_pP_1^p)^\dagger(X_pP_1^p)=P_0^p+P_1^p=\id,
\end{equation}
so Eq.~\eqref{eq:canonical-isometry} is an isometry without assuming local
anticommutation.  Anticommutation will instead be derived on the state.  All
reference matrices are real, so the usual complex-conjugation ambiguity of
self-testing does not alter the target realization considered here.

\section{Self-testing from twenty correlators}
\label{sec:direct}

The construction uses only the two binary observables $Z_p$ and $X_p$ at each
party.  It therefore lies in the tripartite binary-output Bell scenario with
local input cardinalities $(2,2,2)$, the smallest symmetric scenario in which
every party has a nontrivial measurement choice.  Its input is the set of
twenty target correlators summarized in
Table~\ref{tab:direct-data}.  Although the table contains one-, two-, and
three-party moments, all of them are obtained from only five of the eight
possible global measurement contexts,
\begin{equation}
ZZZ,\qquad XZZ,\qquad ZXZ,\qquad ZZX,\qquad XXX.
\label{eq:five-contexts}
\end{equation}
The five listed contexts supply all correlator data used by the proof; the
three contexts containing exactly two $X$ measurements are omitted.

\begin{table}[t]
\caption{The twenty correlators used by the analytic self-test.  In the
middle two rows, $\{p,q,r\}=\{A,B,C\}$.}
\label{tab:direct-data}
\centering
\begin{tabular}{p{0.47\columnwidth}p{0.19\columnwidth}p{0.24\columnwidth}}
\toprule
Correlators & Target value & Contexts \\
\midrule
$Z_S$ for $\emptyset\ne S\subseteq\{A,B,C\}$ & $0$ & $ZZZ$ \\
$X_p,X_pZ_q,X_pZ_r$ & $1/2$ & $XZZ,ZXZ,ZZX$ \\
$X_pZ_qZ_r$ & $-1/2$ & $XZZ,ZXZ,ZZX$ \\
$X_AX_BX_C$ & $1/2$ & $XXX$ \\
\bottomrule
\end{tabular}
\end{table}

\begin{theorem}[Two-setting analytic self-test]
\label{thm:direct-ccz}
Let $\ket\psi$ and $\{Z_p,X_p\}_{p=A,B,C}$ satisfy
Eq.~\eqref{eq:black-box-model}.  If the twenty correlators in
Table~\ref{tab:direct-data} take their target values, then the local isometry
in Eq.~\eqref{eq:canonical-isometry} satisfies
\begin{equation}
\Phi\ket\psi=\Hccz\otimes\ket{\mathrm{junk}}
\label{eq:direct-state-extraction}
\end{equation}
for a normalized auxiliary state.  For every party $p$, the same isometry
also satisfies
\begin{align}
\Phi Z_p\ket\psi
 &=\sigma_z^{(p)}\Hccz\otimes\ket{\mathrm{junk}},\notag\\
\Phi X_p\ket\psi
 &=\sigma_x^{(p)}\Hccz\otimes\ket{\mathrm{junk}}.
\label{eq:direct-measurement-extraction}
\end{align}
Thus the stated correlators self-test the CCZ state and the Pauli $X/Z$
measurements.
\end{theorem}

We give the conceptual proof here and the complete derivation in
Appendix~\ref{app:direct-proof}.  Define the eight computational branches
\begin{equation}
\ket{\psi_{abc}}=P_a^AP_b^BP_c^C\ket\psi.
\label{eq:direct-branches}
\end{equation}
The seven nontrivial $Z$ moments are precisely the nonconstant Fourier
coefficients of the joint $Z$-outcome distribution.  Their vanishing implies
\begin{equation}
\norm{\psi_{abc}}^2=\frac18
\qquad\text{for all }a,b,c\in\{0,1\}.
\label{eq:direct-uniform-branches}
\end{equation}

Next define
\begin{equation}
C_{BC}=\frac{\id+Z_B+Z_C-Z_BZ_C}{2}
=\id-2P_1^BP_1^C
\label{eq:c-bc}
\end{equation}
and cyclic variants.  The four mixed correlators associated with party $A$
give
\begin{equation}
\expect{X_AC_{BC}}=1.
\end{equation}
Because $X_AC_{BC}$ is a Hermitian reflection, saturation gives
\begin{equation}
X_A\ket\psi=C_{BC}\ket\psi,
\label{eq:direct-stabilizer-a}
\end{equation}
and similarly for $B$ and $C$.  These are state-dependent hypergraph
stabilizer relations extracted from observed statistics.

The stabilizer relations do not yet determine how $X_p$ exchanges the two
$Z_p$ sectors.  Let
\begin{equation}
\ket{\eta_{ac}^{(B)}}=P_a^AP_c^C\ket\psi
=\ket{\psi_{a0c}}+\ket{\psi_{a1c}}
\end{equation}
and set $T=\expect{X_AX_BX_C}$.  Resolving $T$ into the four $(a,c)$
sectors and using Eq.~\eqref{eq:direct-stabilizer-a} yields the identity
\begin{equation}
\frac12-T
=\frac12\norm{\{X_B,Z_B\}\ket{\eta_{11}^{(B)}}}^2.
\label{eq:one-square-main}
\end{equation}
The target value $T=1/2$ therefore forces anticommutation on the branch
adjacent to $111$.  Cycling the middle party gives the three top-edge flips.
Cross-party commutation and the equal branch norms then propagate those
relations around the branch cube, without assuming that any branch space is
one-dimensional.  The result is
\begin{align}
X_A\ket{\psi_{abc}}&=(-1)^{bc}\ket{\psi_{1-a,b,c}},\notag\\
X_B\ket{\psi_{abc}}&=(-1)^{ac}\ket{\psi_{a,1-b,c}},\notag\\
X_C\ket{\psi_{abc}}&=(-1)^{ab}\ket{\psi_{a,b,1-c}}.
\label{eq:direct-branch-flips}
\end{align}

\begin{figure}[t]
\centering
\begin{tikzpicture}[
  scale=1.12,
  vertex/.style={circle,draw=black,fill=white,minimum size=6.8mm,inner sep=0pt,font=\scriptsize},
  edge/.style={draw=black!65,line width=0.7pt},
  topedge/.style={draw=quantumviolet,line width=2pt,dash pattern=on 4pt off 2pt},
  every node/.style={font=\scriptsize}
]
\coordinate (000) at (0,0);
\coordinate (100) at (2,0);
\coordinate (010) at (0,2);
\coordinate (110) at (2,2);
\coordinate (001) at (0.8,0.8);
\coordinate (101) at (2.8,0.8);
\coordinate (011) at (0.8,2.8);
\coordinate (111) at (2.8,2.8);
\draw[edge] (000)--(100)--(110)--(010)--cycle;
\draw[edge] (001)--(101)--(111)--(011)--cycle;
\draw[edge] (000)--(001);
\draw[edge] (100)--(101);
\draw[edge] (010)--(011);
\draw[topedge] (111)--(011);
\draw[topedge] (111)--(101);
\draw[topedge] (111)--(110);
\foreach \v in {000,100,010,110,001,101,011,111}
  \node[vertex] at (\v) {$\v$};
\node[quantumviolet,align=center] at (1.4,-0.52)
  {dashed: negative seed flips};
\node[black!65,align=center] at (1.4,-0.80)
  {solid: remaining flips};
\end{tikzpicture}
\caption{The eight $Z$ branches form a cube.  The one-square identity gives
the three negative flips incident on $111$ (violet dashed edges); commutation
and equal branch norms determine the remaining flips.}
\label{fig:branch-cube}
\end{figure}
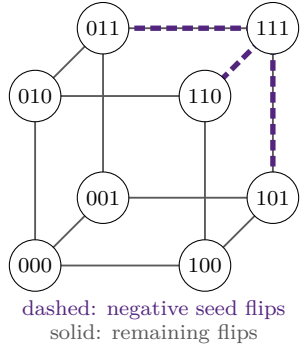

Expanding the local isometry now gives
\begin{equation}
\Phi\ket\psi
=\sum_{a,b,c}\ket{abc}\otimes
X_A^aX_B^bX_C^c\ket{\psi_{abc}}.
\end{equation}
Equation~\eqref{eq:direct-branch-flips} reduces every corrected branch to
\begin{equation}
X_A^aX_B^bX_C^c\ket{\psi_{abc}}
=(-1)^{abc}\ket{\psi_{000}}.
\end{equation}
Since $\norm{\psi_{000}}^2=1/8$, the normalized junk state is
$\ket{\mathrm{junk}}=\sqrt8\ket{\psi_{000}}$, proving
Eq.~\eqref{eq:direct-state-extraction}.  Summing the branch-flip equations
also gives $\{X_p,Z_p\}\ket\psi=0$, from which
Eq.~\eqref{eq:direct-measurement-extraction} follows.

The omitted two-$X$ data are consequences rather than inputs.  For example,
the stabilizer relations and cross-party commutation reduce
$\expect{X_AX_BZ_C^t}$, $t\in\{0,1\}$, to a diagonal $Z$ polynomial.  Uniform
branches then give
\begin{equation}
\expect{X_AX_B}=\expect{X_AX_BZ_C}=\frac12,
\end{equation}
and cyclic permutations supply the other four moments.  Conversely, the
$XXX$ condition plays a genuine role in this proof.  The other nineteen
conditions admit an explicit classical model, given in
Appendix~\ref{app:direct-proof}, with $\expect{XXX}=0$.  This model identifies
the $XXX$ equality as the condition that removes the exhibited
classical component in the present argument.

\begin{corollary}[All two-setting correlators]
\label{cor:full-xz-correlations}
Under the assumptions of Theorem~\ref{thm:direct-ccz}, all twenty-six
nonempty correlators formed from the $X/Z$ measurements are fixed.  The six
correlators not listed in Table~\ref{tab:direct-data} are
\begin{equation}
\expect{X_pX_q}=\expect{X_pX_qZ_r}=\frac12,
\qquad \{p,q,r\}=\{A,B,C\}.
\end{equation}
\end{corollary}

\section{The canonical \texorpdfstring{$X/Z$}{X/Z} measurements cannot
maximize a violated Bell inequality}
\label{sec:obstruction}

Theorem~\ref{thm:direct-ccz} certifies the target from several correlator
constraints.  The same correlations are strictly nonlocal.  To see this,
write $S_X=\sum_pX_p$, $S_{ZZ}=\sum_{p<q}Z_pZ_q$,
$S_{XZ}=\sum_{p\ne q}X_pZ_q$,
$S_{XZZ}=\sum_pX_pZ_qZ_r$, and use the analogous notation for
$S_{XXX}$ and $S_{ZZZ}$.

\begin{proposition}[A five-context Bell violation]
\label{prop:five-context-bell-witness}
Every local correlation satisfies
\begin{equation}
\mathcal I_{20}:=S_{ZZ}+2S_{ZZZ}+S_X+S_{XZ}
-3S_{XZZ}+4S_{XXX}\leq9.
\label{eq:five-context-bell-witness}
\end{equation}
The CCZ state with Pauli $X/Z$ measurements attains
$\mathcal I_{20}=11$.  The expression uses only the five contexts in
Eq.~\eqref{eq:five-contexts}.
\end{proposition}

\begin{proof}
For deterministic outputs, let $k$ be the number of negative values among
$(z_A,z_B,z_C)$.  The part independent of the $x$ outputs is
\begin{equation}
C(z)=z_Az_B+z_Az_C+z_Bz_C+2z_Az_Bz_C,
\end{equation}
and the coefficient of $x_A$ is
$1+z_B+z_C-3z_Bz_C$, with cyclic analogues.  Up to a permutation of the
parties, the pairs consisting of $C(z)$ and the three linear coefficients
are
\begin{equation}
\begin{array}{c|c|c|c|c}
k&0&1&2&3\\ \hline
C(z)&5&-3&1&1\\
(a_A,a_B,a_C)&(0,0,0)&(0,4,4)&(4,4,-4)&(-4,-4,-4).
\end{array}
\end{equation}
Including the term $4x_Ax_Bx_C$, the possible values are respectively
$\{1,9\}$, $\{-15,-7,1,9\}$, $\{-7,9\}$, and $\{-7,9\}$.
Convexity proves the local bound.  The target values
$S_{ZZ}=S_{ZZZ}=0$, $S_X=3/2$, $S_{XZ}=3$,
$S_{XZZ}=-3/2$, and $S_{XXX}=1/2$ give $11$.
\end{proof}

Proposition~\ref{prop:five-context-bell-witness} witnesses strict
nonlocality.  We next ask whether the same canonical realization can
maximize a Bell expression with a strict local--quantum gap.  The theorem
below rules this out even when the expression uses all eight global $X/Z$
contexts, with no assumed permutation symmetry.

\begin{theorem}[No violated Bell inequality is maximized by the canonical
measurements]
\label{thm:xz-obstruction}
Let $p_{\rm CCZ}^{XZ}$ denote the correlations generated by $\Hccz$ with
$Z_p=\sigma_z$ and $X_p=\sigma_x$.  If this implementation attains the
largest quantum value $q$ of a Bell expression $\mathcal I$ in the tripartite
two-input, two-output scenario, then the local bound of $\mathcal I$ is also
$q$.  Equivalently, there exists a local correlation $p_L$ such that
\begin{equation}
\mathcal I(p_L)=\mathcal I(p_{\rm CCZ}^{XZ}).
\label{eq:obstruction-equal-score}
\end{equation}
Consequently, the canonical strategy cannot both attain the quantum maximum
of $\mathcal I$ and yield a value strictly above its local bound.
\end{theorem}

Together with Theorem~\ref{thm:direct-ccz} and
Proposition~\ref{prop:five-context-bell-witness}, this result shows that the
canonical correlations are nonlocal and self-testing from a collection of
correlators, while every Bell expression they maximize has zero
local--quantum gap.  Its convex-geometric content concerns supporting Bell
expressions at the canonical realization and is narrower than non-exposedness
of the correlation point
\cite{GohEtAl2018Geometry,ChenEtAl2023Nonexposed}.

We outline the proof, with all orbit calculations in
Appendix~\ref{app:obstruction-proof}.  First suppose that $\mathcal I$ is
permutation invariant.  Up to a constant it is a linear combination of nine
orbit sums,
\begin{equation}
\mathbf S=(S_X,S_Z,S_{XX},S_{XZ},S_{ZZ},S_{XXX},S_{XXZ},S_{XZZ},S_{ZZZ}).
\end{equation}
Write its coefficient vector as
$\bm\beta=(a,b,c,d,e,f,g,h,i)$.  The orbit vector of the canonical CCZ
correlations is
\begin{equation}
\mathbf p_{\rm CCZ}
=\left(\frac32,0,\frac32,3,0,\frac12,\frac32,-\frac32,0\right).
\label{eq:ccz-xz-orbit-vector}
\end{equation}
If the canonical realization is a quantum maximizer, then fixing the
measurements and varying only the state shows that $\Hccz$ must be a maximal
eigenvector of the associated Bell operator.  Direct expansion in the
permutation-symmetric Hamming-weight basis gives three necessary linear
conditions,
\begin{align}
a-2d-2e+f-h-\frac23i&=0,\notag\\
b+2d+2e-f+g+2h+\frac13i&=0,\notag\\
c-g+\frac23i&=0.
\label{eq:obstruction-eigen-conditions}
\end{align}

Now consider the local correlation obtained by symmetrizing three deterministic
assignments and mixing them with weights $1/4,1/2,1/4$, respectively.  Its
orbit vector is
\begin{equation}
\mathbf p_L
=\left(\frac52,1,2,3,0,\frac12,2,-\frac12,0\right).
\label{eq:obstruction-local-vector}
\end{equation}
The difference obeys
\begin{equation}
\bm\beta\cdot(\mathbf p_L-\mathbf p_{\rm CCZ})
=a+b+\frac12c+\frac12g+h=0,
\label{eq:obstruction-dot-zero}
\end{equation}
where the last equality follows from
Eq.~\eqref{eq:obstruction-eigen-conditions}.  Thus the local and canonical
CCZ correlations have equal value for every permutation-invariant Bell
expression compatible with the required eigenstate condition.

Finally, let $\mathcal I$ be arbitrary and average it over all six party
permutations.  The canonical correlations are symmetric, and the quantum and
local sets are invariant under relabeling.  If the canonical implementation
maximizes $\mathcal I$, it also maximizes the symmetrized expression.  The
local bound of the symmetrized expression cannot exceed that of the original
one.
If $q$ is the canonical quantum value, and $\overline L$ and $L$ are the
local bounds of the averaged and original expressions, respectively, then
Eq.~\eqref{eq:obstruction-dot-zero} and inclusion of the local set in the
quantum set give
\begin{equation}
q\leq\overline L\leq L\leq q.
\label{eq:obstruction-bound-chain}
\end{equation}
Thus both local bounds equal $q$, and a deterministic local maximizer exists
because the local set is a finite polytope.  This proves
Theorem~\ref{thm:xz-obstruction} without assuming symmetry of the original
Bell expression.

Inequivalent or rotated two-setting measurements remain outside
Theorem~\ref{thm:xz-obstruction}.  The next section adds an independent third
measurement and constructs a Bell inequality that is maximally violated by
the target state and measurements.

\section{Self-testing from maximal violation with a third measurement}
\label{sec:bell}

Each party now has three independent reflections $Z_p,X_p,D_p$.  To define
the Bell expression compactly, $S_{\mu_1\cdots\mu_k}$ denotes the sum over
all distinct assignments of the displayed measurements to distinct parties.
For example,
\begin{align}
S_X&=X_A+X_B+X_C,\notag\\
S_{DX}&=D_AX_B+X_AD_B+D_AX_C+X_AD_C+D_BX_C+X_BD_C,\notag\\
S_{XZZ}&=X_AZ_BZ_C+Z_AX_BZ_C+Z_AZ_BX_C.
\end{align}

\subsection{Construction of the Bell inequality}

Bell inequalities tailored to a multipartite target were previously
constructed for the $W$ state by imposing eigenstate and measurement-angle
stationarity conditions, followed by a SWAP-based fidelity analysis
\cite{Pal2014Tomography}.  The construction here starts instead from the
state-dependent relations needed by the SWAP proof.
For words of degree at most two, the target state satisfies the
linear $D$--$X$--$Z$ relation, two ordered $DX$ and $DZ$ relations for each
party, and the three generalized hypergraph stabilizers.  The coefficient
vectors of these relations lie in the kernel of the target word matrix.  We
searched for a positive Gram
form supported on this kernel whose polynomial expansion contains only
permutation-invariant Bell terms.  The search was a feasibility calculation:
it sought a positive-definite Gram matrix in the feasible affine space whose
entries could be reconstructed in $\mathbb Q(\sqrt2)$ and did not optimize
the white-noise threshold.

Permutation symmetry leaves nine orbit types.  A floating-point
semidefinite program located an interior point of the corresponding affine
space; its Bell coefficients and Gram matrix were then reconstructed in
$\mathbb Q(\sqrt2)$ and checked using exact arithmetic.  This
relation-guided construction follows the custom formal-SOS approach of
Ref.~\cite{Barizien2024CustomBell}.

\subsection{Local and quantum bounds}

The nine orbit coefficients and their expanded term counts are shown in
Table~\ref{tab:bell-orbits}.

\begin{table}[H]
\caption{Orbit data for the Bell inequality.  The target value is the
expectation of the complete orbit sum for $\Hccz$ with the reference
measurements in Eqs.~\eqref{eq:reference-xz} and~\eqref{eq:reference-d}.}
\label{tab:bell-orbits}
\centering
\small
\begin{tabular}{lccrr}
\toprule
Orbit & Multiplicity & Coefficient & Target value & Contribution \\
\midrule
$S_X$ & 3 & $12+10\sqrt2$ & $3/2$ & $18+15\sqrt2$ \\
$S_{DX}$ & 6 & $35-6\sqrt2$ & $3\sqrt2$ & $105\sqrt2-36$ \\
$S_{XX}$ & 3 & $-20\sqrt2$ & $3/2$ & $-30\sqrt2$ \\
$S_{XZ}$ & 6 & $15+5\sqrt2$ & $3$ & $45+15\sqrt2$ \\
$S_{DXZ}$ & 6 & $-5$ & $0$ & $0$ \\
$S_{DZZ}$ & 3 & $-20$ & $-3\sqrt2/4$ & $15\sqrt2$ \\
$S_{XXX}$ & 1 & $12$ & $1/2$ & $6$ \\
$S_{XZZ}$ & 3 & $-6$ & $-3/2$ & $9$ \\
$S_{ZZZ}$ & 1 & $18$ & $0$ & $0$ \\
\midrule
Total & 32 & & & $42+120\sqrt2$ \\
\bottomrule
\end{tabular}
\end{table}

Explicitly,
\begin{align}
\Bell={}&(12+10\sqrt2)S_X+(35-6\sqrt2)S_{DX}
-20\sqrt2S_{XX}\notag\\
&+(15+5\sqrt2)S_{XZ}-5S_{DXZ}-20S_{DZZ}\notag\\
&+12S_{XXX}-6S_{XZZ}+18S_{ZZZ}.
\label{eq:final-bell-functional}
\end{align}
Enumeration of the $2^9=512$ deterministic assignments gives
\begin{equation}
\beta_L=258-36\sqrt2.
\label{eq:exact-local-bound}
\end{equation}
The all-$+1$ assignment is the unique deterministic maximizer.  In the
reference realization of Eqs.~\eqref{eq:reference-xz} and
\eqref{eq:reference-d}, direct algebra gives
\begin{equation}
\Bell\Hccz=(42+120\sqrt2)\Hccz.
\label{eq:target-bell-eigenstate}
\end{equation}

\begin{theorem}[Bell self-test at maximal violation]
\label{thm:bell-ccz}
In the three-party reflection model of Eq.~\eqref{eq:black-box-model}, the
quantum maximum of Eq.~\eqref{eq:final-bell-functional} is
\begin{equation}
\beta_Q=42+120\sqrt2.
\label{eq:exact-quantum-bound}
\end{equation}
If a realization attains this value, the isometry
Eq.~\eqref{eq:canonical-isometry} extracts the CCZ state and reproduces the
action of the three reference measurements on the state:
\begin{align}
\Phi\ket\psi&=\Hccz\otimes\ket{\mathrm{junk}},\notag\\
\Phi Z_p\ket\psi&=\sigma_z^{(p)}\Hccz\otimes\ket{\mathrm{junk}},\notag\\
\Phi X_p\ket\psi&=\sigma_x^{(p)}\Hccz\otimes\ket{\mathrm{junk}},\notag\\
\Phi D_p\ket\psi&=\frac{\sigma_x^{(p)}+\sigma_z^{(p)}}{\sqrt2}
\Hccz\otimes\ket{\mathrm{junk}}.
\label{eq:bell-measurement-extraction}
\end{align}
\end{theorem}
\FloatBarrier

\subsection{Equality conditions and state extraction}

The Bell gap is strictly positive,
\begin{equation}
\beta_Q-\beta_L=156\sqrt2-216>0.
\end{equation}
The upper bound in Theorem~\ref{thm:bell-ccz} is established by an exact
computer-assisted certificate.  Let $v_2$ be
the column of all reduced noncommutative words of degree at most two.  There
are 55 such words.  Their ideal action on the unnormalized CCZ vector defines
an $8\times55$ matrix $W_\star$ over $\field$ with
\begin{equation}
\rank W_\star=8,\qquad \dim\ker W_\star=47.
\label{eq:ideal-word-rank-main}
\end{equation}
Choose a full-column-rank matrix
$R\in\field^{55\times47}$ whose range is this kernel.  The certificate
contains a positive-definite matrix $S\in\field^{47\times47}$ such that the
formal identity
\begin{equation}
\beta_Q\id-\Bell=v_2^\dagger RSR^Tv_2
\label{eq:exact-sos-main}
\end{equation}
holds in the quotient algebra generated by the reflection and cross-party
commutation relations.  All 784 canonical-word coefficients are matched
without approximation.  Positive definiteness of $S$ is certified exactly
by a rational congruence and a Gershgorin bound.

Saturation of Eq.~\eqref{eq:exact-sos-main} forces the complete target
relation space of degree at most two on the unknown state.  In particular,
for each party $p$ it implies the three state-dependent relations
\begin{align}
(\sqrt2D_p-X_p-Z_p)\ket\psi&=0,
\label{eq:calibration-main}\\
(\sqrt2D_pX_p-\id-Z_pX_p)\ket\psi&=0,
\label{eq:dx-main}\\
(\sqrt2D_pZ_p-X_pZ_p-\id)\ket\psi&=0,
\label{eq:dz-main}
\end{align}
as well as the hypergraph stabilizer relations
\begin{equation}
(2X_p-\id-Z_q-Z_r+Z_qZ_r)\ket\psi=0,
\qquad \{p,q,r\}=\{A,B,C\}.
\label{eq:bell-stabilizer-main}
\end{equation}
Equation~\eqref{eq:calibration-main} is a relation on $\ket\psi$, not a
device-side operator identity.  It cannot be multiplied on
the right by an unknown same-party operator.  Instead, left-multiplying by
$D_p$ and independently using Eqs.~\eqref{eq:dx-main}--\eqref{eq:dz-main}
gives the legal state-dependent consequence
\begin{equation}
\{X_p,Z_p\}\ket\psi=0.
\label{eq:bell-anticommutation-main}
\end{equation}

Equations~\eqref{eq:bell-stabilizer-main} and
\eqref{eq:bell-anticommutation-main} yield the same branch-flip rules as
Eq.~\eqref{eq:direct-branch-flips}.  Here their unitarity also forces all
branch norms to be equal.  The canonical isometry therefore extracts
$\Hccz$, and Eq.~\eqref{eq:bell-measurement-extraction} follows from the
first relation above.  Appendices~\ref{app:bell-algebra}--
\ref{app:equality-swap} separate the quotient algebra, positivity
verification, and analytic equality-to-isometry argument.  The SOS identity
proves the quantum upper bound by exact computer-assisted algebra, while the
derivation of the SWAP isometry from its equality conditions is analytic.

\begin{table}[t]
\caption{Dimensions entering the SOS proof and the certified spectral bound.}
\label{tab:certificate-audit}
\centering
\begin{tabular}{lr}
\toprule
Quantity & Value \\
\midrule
Dimension of $v_2$ & 55 \\
$\rank W_\star$ & 8 \\
$\dim\ker W_\star$ & 47 \\
Canonical-word coefficients & 784 \\
Rigorous lower bound on $\lambda_{\min}(S)$ & $>1/7000$ \\
\bottomrule
\end{tabular}
\end{table}
\FloatBarrier

The SOS proof follows the general strategy of extracting self-testing
relations from a Bell operator \cite{BampsPironio2015SOS}.  The present
construction uses a finite algebraic certificate in the custom formal-SOS
framework of Ref.~\cite{Barizien2024CustomBell}.

\section{Bounds for imperfect data}
\label{sec:quantitative}

The following estimates quantify deviations of selected operator relations
away from ideal data.  A fidelity bound for the extracted state would
additionally require simultaneous control of all noisy correlators.

For the analytic proof, suppose that the seven pure-$Z$ moments and the twelve
single-$X$ moments in Table~\ref{tab:direct-data} retain their target values,
while
\begin{equation}
T=\expect{X_AX_BX_C}=\frac12-\delta.
\end{equation}
The norm-square identity in Eq.~\eqref{eq:one-square-main} then gives, for
each party $p$ and the other two parties $q,r$,
\begin{equation}
\norm{\{X_p,Z_p\}P_1^qP_1^r\ket\psi}=\sqrt{2\delta}.
\label{eq:direct-local-defect}
\end{equation}
This identity gives the local defect in the doubly conditioned sector.
Extending Eq.~\eqref{eq:direct-local-defect} to jointly noisy data requires
simultaneous control of the branch weights, stabilizer relations, and cube
propagation.

For the Bell inequality, set
\begin{equation}
\varepsilon:=\beta_Q-\expect{\Bell}\geq0,
\qquad
\ket y:=R^Tv_2\ket\psi.
\end{equation}
Taking the expectation of the SOS identity gives
\begin{equation}
\varepsilon=\bra yS\ket y
\geq\lambda_{\min}(S)\norm{y}^2,
\end{equation}
and hence
\begin{equation}
\norm{R^Tv_2\ket\psi}
\leq\sqrt{\frac{\varepsilon}{\lambda_{\min}(S)}}.
\label{eq:relation-space-defect}
\end{equation}
Here $\lambda_{\min}(S)>0$ depends on the certificate.  The rational
congruence in
Appendix~\ref{app:exact-sos} gives $\lambda_{\min}(S)>1/7000$, and therefore
\begin{equation}
\norm{R^Tv_2\ket\psi}\leq\sqrt{7000\,\varepsilon}.
\label{eq:relation-space-defect-explicit}
\end{equation}
Equation~\eqref{eq:relation-space-defect-explicit} controls the complete ideal
degree-two relation vector.  Turning this bound into a robust extraction
theorem requires propagating the vector errors through conditioning, branch
normalization, and the local isometry.

Every term in the target Bell operator is traceless.  With the trusted target
measurements and the white-noise family
\begin{equation}
\rho_v=v\ket{H_3}\!\bra{H_3}+(1-v)\frac{\id}{8},
\end{equation}
the Bell value is $v\beta_Q$.  Strict violation of the local bound therefore
requires
\begin{equation}
v>\frac{\beta_L}{\beta_Q}
=\frac{258-36\sqrt2}{42+120\sqrt2}
\simeq0.97819.
\label{eq:white-noise-threshold}
\end{equation}
Equation~\eqref{eq:white-noise-threshold} gives the white-noise threshold for
Bell violation in this model.  Its magnitude reflects the
certificate-oriented construction; noise tolerance was not an optimization
objective.

\section{Discussion}
\label{sec:discussion}

The CCZ state separates two certification tasks: self-testing from a
collection of correlator equalities and self-testing at the maximum of a
single Bell inequality.  Twenty $X/Z$ correlators drawn from five global
contexts accomplish the first task with two binary inputs per party.  These
correlations are nonlocal, yet Theorem~\ref{thm:xz-obstruction} shows that
every Bell expression maximized by the canonical measurements has zero
local--quantum gap.  The third input used in Theorem~\ref{thm:bell-ccz}
provides additional state-dependent relations and yields both a strict gap
and self-testing from maximal violation.

The branch-cube argument gives a direct explanation of how the non-Pauli
phase is recovered.  For the Boolean phase $f(a,b,c)=abc$, flipping the first
bit changes the sign by $(-1)^{bc}$, with cyclic analogues for the other two
bits.  The generalized stabilizers determine these conditional flip signs,
while the $XXX$ equality fixes the three negative edges incident on $111$.
Commutation between different parties then propagates the remaining edges.
This mechanism differs from a graph-state proof, where Pauli stabilizers
already contain all required signs as linear operator relations.

The two proofs use different information.  The analytic $X/Z$ argument
reconstructs the cubic phase from conditional bit-flip signs, without
assuming qubit dimension or one-dimensional computational branches.  The
Bell construction instead promotes the state-dependent relations needed for
extraction to equality conditions of one Bell inequality.
Theorem~\ref{thm:xz-obstruction} explains why the latter task cannot use the
canonical $X/Z$ realization.  Relative to universal multipartite
constructions such as Ref.~\cite{BalanzoJuando2026AllPure}, specializing to
CCZ trades generality for two inputs per party, five global contexts, and a
direct account of how the cubic phase is recovered.

These mechanisms suggest a question for phase-polynomial states
\begin{equation}
\ket{\psi_f}=2^{-n/2}\sum_{x\in\{0,1\}^n}(-1)^{f(x)}\ket{x}.
\end{equation}
An analogous proof would need uniform computational branches, conditional
relations encoding the Boolean derivatives of $f$, enough phase-sensitive
seed edges, and a connected propagation graph.  Whether these conditions
suffice beyond the three-qubit CCZ state remains open.  A
relation-guided SOS construction for a larger state would also have to
produce a Bell operator with a strict local--quantum gap and an equality space
strong enough for extraction.

The CCZ state is both the smallest rank-three hypergraph state and an
elementary entangled magic resource, making it a concrete benchmark for
device-independent certification beyond stabilizer states.  Several resource
and robustness questions remain.  We do not know whether fewer than twenty
correlators suffice.  Theorem~\ref{thm:xz-obstruction} applies only to the
canonical Pauli $X/Z$ measurements and leaves inequivalent two-setting
realizations unresolved.  The final Bell inequality has a white-noise
violation threshold of approximately $97.819\%$, and a fidelity bound for
jointly noisy data remains to be derived.  Improving the noise tolerance,
propagating the relation-space estimates to an extracted-state bound, and
identifying when branch propagation extends to higher-degree phase
polynomials are natural next steps.  Taken together, the results show that
the CCZ cubic phase can be certified device-independently with two local
inputs per party, even though the canonical $X/Z$ realization cannot attain
a strictly Bell-violating quantum maximum; an explicit third-input
construction achieves self-testing from maximal violation.

\section*{Acknowledgments}

This work was supported by the National Natural Science Foundation of China
(Grant No.~62201252), the Fundamental Research Funds for the Central
Universities (Grant No.~NS2025030), and the Beijing Natural Science Foundation
(Grant No.~4262015).

\section*{Code and data availability}

The code, certificates, and reproduction instructions are available at
\href{https://github.com/brant-fudan/self-testing-ccz-state}
{github.com/brant-fudan/self-testing-ccz-state}.

\section*{Author contributions}

Y.H. conceived and supervised the project, developed the methodology, and led
the analysis and writing.  X.K. and X.B. contributed to formal analysis,
exact and numerical calculations, validation, and manuscript revision; X.B.
also prepared the visualizations.  Y.W. contributed to conceptual development,
methodology, validation, supervision, and manuscript revision.  All authors
reviewed and approved the manuscript.

\clearpage

\appendix
\section{Proof of the twenty-correlator self-test}
\label{app:direct-proof}

This appendix gives an analytic proof of Theorem~\ref{thm:direct-ccz}.
Throughout, $\ket\psi$ is normalized and the unknown $Z_p,X_p$ satisfy
Eq.~\eqref{eq:black-box-model}.

\subsection{Uniform computational branches}

Expanding the projectors in Eq.~\eqref{eq:direct-branches} gives
\begin{equation}
\norm{\psi_{abc}}^2
=\frac18\sum_{S\subseteq\{A,B,C\}}
(-1)^{\sum_{p\in S}a_p}\expect{\prod_{p\in S}Z_p},
\label{eq:branch-fourier}
\end{equation}
where the empty product contributes one.  The first row of
Table~\ref{tab:direct-data} sets all seven nonconstant Fourier coefficients
to zero, proving Eq.~\eqref{eq:direct-uniform-branches}.  We write
$r=1/\sqrt8$ for the common branch norm.

\subsection{Saturated generalized stabilizers}

The operator $C_{BC}=\id-2P_1^BP_1^C$ is a Hermitian reflection.  Since it
commutes with $X_A$, so is $S_A=X_AC_{BC}$.  The four moments with one $X_A$
give
\begin{align}
\expect{S_A}
&=\frac12\left(
\expect{X_A}+\expect{X_AZ_B}+\expect{X_AZ_C}
-\expect{X_AZ_BZ_C}\right)=1.
\end{align}
For a reflection, unit expectation is equivalent to zero variance; hence
$S_A\ket\psi=\ket\psi$ and
Eq.~\eqref{eq:direct-stabilizer-a} follows.  Cycling the parties gives
\begin{equation}
X_A\ket\psi=C_{BC}\ket\psi,\quad
X_B\ket\psi=C_{AC}\ket\psi,\quad
X_C\ket\psi=C_{AB}\ket\psi.
\label{eq:all-direct-stabilizers}
\end{equation}
Projecting the first relation onto fixed $B,C$ sectors gives
\begin{equation}
X_A(\ket{\psi_{0bc}}+\ket{\psi_{1bc}})
=(-1)^{bc}(\ket{\psi_{0bc}}+\ket{\psi_{1bc}}),
\label{eq:sector-stabilizer-a}
\end{equation}
with cyclic analogues.

\subsection{The one-square identity}

For the middle party $B$, define
\begin{equation}
\ket{\eta_{ac}^{(B)}}=P_a^AP_c^C\ket\psi.
\end{equation}
The $B$ stabilizer in Eq.~\eqref{eq:all-direct-stabilizers} and the uniform
branches give
\begin{equation}
X_B\ket{\eta_{ac}^{(B)}}=(-1)^{ac}\ket{\eta_{ac}^{(B)}},
\quad
\norm{\eta_{ac}^{(B)}}^2=\frac14,
\quad
\bra{\eta_{ac}^{(B)}}Z_B\ket{\eta_{ac}^{(B)}}=0.
\label{eq:eta-properties}
\end{equation}
In $T=\expect{X_AX_BX_C}$, use the $A$ relation on the bra and the $C$
relation on the ket.  Since $C_{BC}=Z_B^c$ in the $C=c$ sector and
$C_{AB}=Z_B^a$ in the $A=a$ sector,
\begin{equation}
T=\sum_{a,c\in\{0,1\}}
\bra{\eta_{ac}^{(B)}}Z_B^cX_BZ_B^a\ket{\eta_{ac}^{(B)}}.
\label{eq:t-sector-sum}
\end{equation}
The four terms are $1/4,0,0$, and
$q_B=\bra{\eta_{11}^{(B)}}Z_BX_BZ_B\ket{\eta_{11}^{(B)}}$.  Set
$R_B=Z_BX_BZ_B$.  Because $R_B$ is a reflection and
$X_B\ket{\eta_{11}^{(B)}}=-\ket{\eta_{11}^{(B)}}$,
\begin{align}
\frac12-T
&=\frac12\norm{(\id-R_B)\ket{\eta_{11}^{(B)}}}^2\notag\\
&=\frac12\norm{\{X_B,Z_B\}\ket{\eta_{11}^{(B)}}}^2,
\end{align}
which is Eq.~\eqref{eq:one-square-main}.  At $T=1/2$ the norm vanishes.
Resolving the two $B$ branches yields
\begin{equation}
X_B\ket{\psi_{101}}=-\ket{\psi_{111}},\qquad
X_B\ket{\psi_{111}}=-\ket{\psi_{101}}.
\label{eq:b-top-edge}
\end{equation}
Using $A$ and $C$ as the middle party gives
\begin{align}
X_A\ket{\psi_{011}}&=-\ket{\psi_{111}},&
X_A\ket{\psi_{111}}&=-\ket{\psi_{011}},\notag\\
X_C\ket{\psi_{110}}&=-\ket{\psi_{111}},&
X_C\ket{\psi_{111}}&=-\ket{\psi_{110}}.
\label{eq:other-top-edges}
\end{align}

\subsection{Propagation over the branch cube}

Consider the face $c=1$.  Cross-party commutation applied to the top vertex
gives
\begin{equation}
X_A\ket{\psi_{101}}=X_B\ket{\psi_{011}}
=:\ket{w_{001}}.
\label{eq:w001}
\end{equation}
The first expression places $\ket{w_{001}}$ in the $B=0,C=1$ sector and the
second in the $A=0,C=1$ sector.  Hence it lies in the $001$ branch and has
norm $r$.  Equation~\eqref{eq:sector-stabilizer-a} in the $bc=01$ sector
gives
\begin{equation}
X_A\ket{\psi_{001}}
=\ket{\psi_{101}}+(\ket{\psi_{001}}-\ket{w_{001}}).
\end{equation}
The two terms on the right have orthogonal $A$ support.  Unitarity of $X_A$
and the equal branch norms therefore imply
\begin{equation}
r^2=r^2+\norm{\psi_{001}-w_{001}}^2,
\end{equation}
so $\ket{w_{001}}=\ket{\psi_{001}}$.  Repeating this argument on the $b=1$
and $a=1$ faces gives
\begin{align}
X_A\ket{\psi_{101}}&=\ket{\psi_{001}},&
X_B\ket{\psi_{011}}&=\ket{\psi_{001}},\notag\\
X_A\ket{\psi_{110}}&=\ket{\psi_{010}},&
X_C\ket{\psi_{011}}&=\ket{\psi_{010}},\notag\\
X_B\ket{\psi_{110}}&=\ket{\psi_{100}},&
X_C\ket{\psi_{101}}&=\ket{\psi_{100}},
\label{eq:middle-edge-propagation}
\end{align}
together with the reverse directions.

For the bottom face, commutation on $\ket{\psi_{100}}$ gives
$X_A\ket{\psi_{100}}=X_B\ket{\psi_{010}}=: \ket{w_{000}}$.  The same
support and norm argument identifies $\ket{w_{000}}=\ket{\psi_{000}}$.
The last edge follows explicitly from
\begin{equation}
X_C\ket{\psi_{001}}
=X_CX_A\ket{\psi_{101}}
=X_AX_C\ket{\psi_{101}}
=X_A\ket{\psi_{100}}
=\ket{\psi_{000}}.
\end{equation}
This proves all three relations in
Eq.~\eqref{eq:direct-branch-flips}.  Adding them over the eight
branches yields
\begin{equation}
\{X_p,Z_p\}\ket\psi=0\qquad(p=A,B,C).
\label{eq:direct-global-anti}
\end{equation}

\subsection{State and measurement extraction}

Using Eq.~\eqref{eq:canonical-isometry} for all parties gives
\begin{equation}
\Phi\ket\psi=\sum_{a,b,c}\ket{abc}\otimes
X_A^aX_B^bX_C^c\ket{\psi_{abc}}.
\end{equation}
The branch flips reduce the corrected vector to
$(-1)^{abc}\ket{\psi_{000}}$, so
\begin{equation}
\Phi\ket\psi
=\frac1{\sqrt8}\sum_{a,b,c}(-1)^{abc}\ket{abc}
\otimes\sqrt8\ket{\psi_{000}}.
\end{equation}
This is Eq.~\eqref{eq:direct-state-extraction}.  The $Z$ action follows
directly from $Z_pP_a^p=(-1)^aP_a^p$.  Equation~\eqref{eq:direct-global-anti}
implies
\begin{equation}
P_0^pX_p\ket\psi=X_pP_1^p\ket\psi,\qquad
P_1^pX_p\ket\psi=X_pP_0^p\ket\psi,
\end{equation}
which yields the $X$ line of
Eq.~\eqref{eq:direct-measurement-extraction}.

\subsection{Redundant moments and the role of
\texorpdfstring{$XXX$}{XXX}}

For distinct parties $A,B$ and $t\in\{0,1\}$, the saturated stabilizers and
cross-party commutation give
\begin{equation}
\expect{X_AX_BZ_C^t}=\expect{C_{BC}C_{AC}Z_C^t}.
\end{equation}
The right-hand side is diagonal in the uniform $Z$ distribution, giving
\begin{equation}
\expect{X_AX_B}=\expect{X_AX_BZ_C}=\frac12.
\end{equation}
Cyclic permutations show that the
three $XX$ and three $XXZ$ moments omitted from Table~\ref{tab:direct-data}
are consequences of the assumed data.

To see why the final $XXX$ value matters, choose a uniformly random hidden
bit string $(a,b,c)$ and define deterministic outputs
\begin{equation}
z_A=(-1)^a,\quad z_B=(-1)^b,\quad z_C=(-1)^c,
\end{equation}
\begin{equation}
x_A=(-1)^{bc},\quad x_B=(-1)^{ac},\quad x_C=(-1)^{ab}.
\end{equation}
This local model satisfies the seven pure-$Z$ moments and all twelve
single-$X$ mixed moments, but gives $\expect{X_AX_BX_C}=0$.  Thus the $XXX$
condition excludes a local correlation compatible with the other nineteen
equations.

\section{Proof of Theorem~\ref{thm:xz-obstruction}}
\label{app:obstruction-proof}

We give the orbit calculation behind
Theorem~\ref{thm:xz-obstruction}.  For a permutation-invariant Bell expression,
define
\begin{align}
S_X&=X_A+X_B+X_C,&
S_Z&=Z_A+Z_B+Z_C,\notag\\
S_{XX}&=X_AX_B+X_AX_C+X_BX_C,&
S_{ZZ}&=Z_AZ_B+Z_AZ_C+Z_BZ_C,\notag\\
S_{XZ}&=\sum_{p\ne q}X_pZ_q,&
S_{XXX}&=X_AX_BX_C,\notag\\
S_{XXZ}&=X_AX_BZ_C+X_AX_CZ_B+X_BX_CZ_A,\notag\\
S_{XZZ}&=X_AZ_BZ_C+X_BZ_AZ_C+X_CZ_AZ_B,&
S_{ZZZ}&=Z_AZ_BZ_C.
\label{eq:xz-nine-orbits}
\end{align}
Every such two-setting binary-output Bell expression is, up to an irrelevant
constant,
\begin{equation}
\mathcal I_{\bm\beta}=aS_X+bS_Z+cS_{XX}+dS_{XZ}+eS_{ZZ}
+fS_{XXX}+gS_{XXZ}+hS_{XZZ}+iS_{ZZZ}.
\label{eq:pi-xz-functional}
\end{equation}
Probability-form terms introduce no additional freedom: binary probabilities
have a correlator Fourier expansion and nonsignaling identifies a marginal
across contexts.

Let $\ket{s_w}$ be the unnormalized sum of computational basis vectors of
Hamming weight $w$.  Then
\begin{equation}
\sqrt8\Hccz=\ket{s_0}+\ket{s_1}+\ket{s_2}-\ket{s_3}.
\end{equation}
The action of the nine orbit operators on this vector, expressed in the
ordered weight basis, is shown in Table~\ref{tab:weight-action}.

\begin{table*}[t]
\caption{Coefficients of each orbit operator acting on
$\sqrt8\Hccz$, in the basis
$(\ket{s_0},\ket{s_1},\ket{s_2},\ket{s_3})$.}
\label{tab:weight-action}
\centering
\begin{tabular}{lrrrrrrrrr}
\toprule
&$S_X$&$S_Z$&$S_{XX}$&$S_{XZ}$&$S_{ZZ}$&$S_{XXX}$&$S_{XXZ}$&$S_{XZZ}$&$S_{ZZZ}$\\
\midrule
$s_0$&3&3&3&6&3&$-1$&3&3&1\\
$s_1$&3&1&1&2&$-1$&1&3&$-1$&$-1$\\
$s_2$&1&$-1$&3&2&$-1$&1&$-1$&$-3$&1\\
$s_3$&3&3&3&$-6$&$-3$&1&$-3$&3&1\\
\bottomrule
\end{tabular}
\end{table*}

If the canonical implementation attains the largest quantum value, then
$\Hccz$ must be an eigenvector of its fixed Bell operator.  If the resulting
weight coefficients are $(r_0,r_1,r_2,r_3)$, the eigenvector condition is
$r_0=r_1=r_2=-r_3$.  Eliminating the four coefficients gives
Eq.~\eqref{eq:obstruction-eigen-conditions}, or equivalently
\begin{align}
a&=2d+2e-f+h+\frac23i,\notag\\
b&=-2d-2e+f-g-2h-\frac13i,\notag\\
c&=g-\frac23i.
\label{eq:obstruction-eliminated}
\end{align}
The canonical CCZ value is then $6d+3e-f+3g$.

The local correlation in
Eq.~\eqref{eq:obstruction-local-vector} is built from the three deterministic
assignments
\begin{center}
\begin{tabular}{c|c|c}
\toprule
&$(z_A,z_B,z_C)$&$(x_A,x_B,x_C)$\\
\midrule
$\lambda_1$&$(-1,-1,+1)$&$(+1,+1,-1)$\\
$\lambda_2$&$(-1,+1,+1)$&$(+1,+1,+1)$\\
$\lambda_3$&$(+1,+1,+1)$&$(+1,+1,+1)$\\
\bottomrule
\end{tabular}
\end{center}
Let $\operatorname{Sym}(\lambda)$ denote the uniform mixture over all six
party permutations, with repeated deterministic assignments counted with
their group multiplicity.  Then
\begin{equation}
p_L=\frac14\operatorname{Sym}(\lambda_1)
+\frac12\operatorname{Sym}(\lambda_2)
+\frac14\operatorname{Sym}(\lambda_3)
\end{equation}
has the orbit vector in Eq.~\eqref{eq:obstruction-local-vector}.  Substitution
of Eq.~\eqref{eq:obstruction-eliminated} proves
Eq.~\eqref{eq:obstruction-dot-zero}.  Hence a local correlation
attains the same value as the canonical strategy for every
permutation-invariant Bell expression that can have $\Hccz$ as an optimizing
eigenstate.

For a general Bell expression $\mathcal I$, define
\begin{equation}
\overline{\mathcal I}=\frac16\sum_{\pi\in S_3}\pi(\mathcal I).
\end{equation}
If the symmetric canonical correlations maximize $\mathcal I$ with value $q$,
it also maximizes $\overline{\mathcal I}$ with value $q$.  If $L$ is the
local bound of $\mathcal I$, invariance of the local set gives
$\overline L\le L$ for the local bound of the average.  The preceding
permutation-invariant argument gives $\overline L\ge q$, and therefore
$L\ge q$.  Conversely, the local set is contained in the quantum set, so
$L\le q$.  Hence $\overline L=L=q$, and the maximum over the finite local
polytope is attained by a deterministic local correlation.  This completes
the proof for arbitrary linear Bell expressions.

\section{Bell algebra, local bound, and target values}
\label{app:bell-algebra}

The algebra for Theorem~\ref{thm:bell-ccz} is generated by the nine letters
$Z_p,X_p,D_p$.  Within one party, only adjacent repetitions are removed using
$M_p^2=\id$; distinct same-party letters are neither commuted nor
anticommuted.  Letters of different parties commute and are placed in the
fixed party order $A,B,C$.  This defines a canonical reduced word for every
monomial and prevents the reference Pauli algebra from being inserted into
the black-box problem.

The orbit convention in Section~\ref{sec:bell} expands as
\begin{align}
S_X={}&X_A+X_B+X_C,\notag\\
S_{DX}={}&D_AX_B+D_AX_C+D_BX_A+D_BX_C+D_CX_A+D_CX_B,\notag\\
S_{XX}={}&X_AX_B+X_AX_C+X_BX_C,\notag\\
S_{XZ}={}&X_AZ_B+X_AZ_C+X_BZ_A+X_BZ_C+X_CZ_A+X_CZ_B,\notag\\
S_{DXZ}={}&D_AX_BZ_C+D_AX_CZ_B+D_BX_AZ_C
+D_BX_CZ_A\notag\\
&+D_CX_AZ_B+D_CX_BZ_A,\notag\\
S_{DZZ}={}&D_AZ_BZ_C+Z_AD_BZ_C+Z_AZ_BD_C,\notag\\
S_{XXX}={}&X_AX_BX_C,\notag\\
S_{XZZ}={}&X_AZ_BZ_C+Z_AX_BZ_C+Z_AZ_BX_C,\notag\\
S_{ZZZ}={}&Z_AZ_BZ_C.
\label{eq:bell-orbits-expanded}
\end{align}
This gives the 32 distinct terms reported in
Table~\ref{tab:bell-orbits}.

For a deterministic local strategy, each of the nine measurements has a
preassigned output in $\{\pm1\}$.  Substitution into
Eq.~\eqref{eq:final-bell-functional} produces an element of $\field$.
Finite enumeration of all 512 assignments reduces the local-bound statement
to the inequalities
\begin{equation}
\beta_L-B_\lambda\geq0
\qquad(\lambda\in\{\pm1\}^9).
\end{equation}
Equality holds only for the all-$+1$ assignment.  The accompanying verifier
checks all 512 inequalities directly with exact comparisons of numbers
$a+b\sqrt2$ and confirms the unique equality case.

For the target evaluation, substitute
\begin{equation}
Z=\begin{pmatrix}1&0\\0&-1\end{pmatrix},\qquad
X=\begin{pmatrix}0&1\\1&0\end{pmatrix},\qquad
D=\frac{X+Z}{\sqrt2}
\end{equation}
in the trusted reference representation.  Direct Kronecker products give
Eq.~\eqref{eq:target-bell-eigenstate}.  This computation proves attainment;
Appendix~\ref{app:exact-sos} proves the black-box upper bound.

The reference correlators can all be obtained directly from
Eq.~\eqref{eq:ccz-state-intro}.  For the trusted observables
\begin{equation}
Z=\sigma_z,
\qquad X=\sigma_x,
\qquad D=\frac{Z+X}{\sqrt2},
\end{equation}
the one-body values are
\begin{equation}
\expect Z=0,
\qquad
\expect X=\frac12,
\qquad
\expect D=\frac{\sqrt2}{4}.
\label{eq:one-body-reference-table}
\end{equation}
For any ordered pair of distinct parties, permutation symmetry gives the
two-body table
\begin{equation}
\begin{array}{c|ccc}
 & Z & X & D \\ \hline
Z & 0 & \frac12 & \frac{\sqrt2}{4} \\
X & \frac12 & \frac12 & \frac{\sqrt2}{2} \\
D & \frac{\sqrt2}{4} & \frac{\sqrt2}{2} & \frac34
\end{array}.
\label{eq:two-body-reference-table}
\end{equation}
The three-body value depends only on the multiset of settings:
\begin{equation}
\begin{array}{c|c@{\qquad}c|c}
ZZZ & 0 & ZZX & -\frac12 \\
ZZD & -\frac{\sqrt2}{4} & ZXX & \frac12 \\
ZXD & 0 & ZDD & -\frac14 \\
XXX & \frac12 & XXD & \frac{\sqrt2}{2} \\
XDD & \frac12 & DDD & \frac{\sqrt2}{8}
\end{array}.
\label{eq:three-body-reference-table}
\end{equation}
Every permutation of a displayed multiset has the same value.  These data
reproduce every ideal orbit value and contribution in
Table~\ref{tab:bell-orbits}.  In particular, the $X/Z$ entries recover the
twenty inputs of Theorem~\ref{thm:direct-ccz}.

\section{Sum-of-squares proof of the quantum bound}
\label{app:exact-sos}

This appendix specifies the finite algebraic object that proves the quantum
upper bound.  Let $v_2$ contain all canonical words of total degree at most
two.  There is one degree-zero word, nine degree-one words, eighteen ordered
same-party degree-two words, and twenty-seven cross-party products.  Hence
\begin{equation}
\lvert v_2\rvert=1+9+18+27=55.
\end{equation}

Let $\ket{\widetilde H_3}=\sqrt8\Hccz$ and define the ideal word-evaluation
matrix by
\begin{equation}
(W_\star)_j=v_{2,j}\ket{\widetilde H_3},
\qquad W_\star\in\field^{8\times55},
\end{equation}
where the trusted reference measurements are used only for this ideal
evaluation.  Row reduction over $\mathbb Q(\sqrt2)$ gives
\begin{equation}
\rank W_\star=8,\qquad \dim\ker W_\star=47.
\end{equation}
Selecting an invertible $8\times8$ pivot submatrix and reducing the
remaining columns yields a full-column-rank matrix
\begin{equation}
R\in\field^{55\times47},\qquad W_\star R=0,
\qquad \operatorname{range}R=\ker W_\star.
\label{eq:exact-r-kernel}
\end{equation}

The certificate stores the 1128 upper-triangular entries of a symmetric
$47\times47$ matrix $S$ over $\field$.  A verifier reconstructs the full
matrix and forms $G=RSR^T$.  Expanding $v_2^\dagger Gv_2$, reducing all
ordered word products, and collecting equal terms produces 784 canonical
words.  For each canonical word, the verifier checks in
$\mathbb Q(\sqrt2)$ that its coefficient equals the corresponding coefficient
of $\beta_Q\id-\Bell$.  The resulting residual is identically zero, proving the
formal identity in Eq.~\eqref{eq:exact-sos-main}.

Positive definiteness is certified exactly.  The certificate includes an
invertible rational upper-triangular matrix $P$.  Set
\begin{equation}
T=P^TSP,\qquad E=T-\id.
\end{equation}
Rational interval arithmetic enclosing $\sqrt2$, independently repeated with
symbolic pairs $a+b\sqrt2$, verifies
\begin{equation}
\beta=\max_i\sum_j\lvert E_{ij}\rvert<1.
\label{eq:gershgorin-bound-app}
\end{equation}
The displayed decimal
$\beta\simeq3.359534392389232\times10^{-6}$ is a readable approximation; the
strict comparison in
Eq.~\eqref{eq:gershgorin-bound-app} uses rational and integer inequalities.
Gershgorin's theorem gives $T\succ0$, and congruence by invertible $P$ gives
$S\succ0$.  Therefore $G\succeq0$, and
Eq.~\eqref{eq:exact-sos-main} proves $\Bell\le\beta_Q$ for every black-box
reflection realization.  Target attainment then proves the claimed
quantum maximum.

The same data give a rigorous spectral constant.  For $x=Py$,
\begin{equation}
x^TSx=y^TTy\geq(1-\beta)\norm{y}^2
\geq\frac{1-\beta}{\norm{P}_F^2}\norm{x}^2.
\end{equation}
The rational values stored in the certificate satisfy
\begin{equation}
\frac{1-\beta}{\norm{P}_F^2}>\frac1{7000},
\label{eq:lambda-min-certified}
\end{equation}
so $\lambda_{\min}(S)>1/7000$.  The verifier checks this rational inequality
exactly.

The numerical SDP supplies an interior candidate.  The verified algebraic
identity and exact positivity proof establish the theorem; floating affine
residuals and Gram eigenvalues are retained as diagnostics.

Suppose now that $\expect{\Bell}=\beta_Q$ and define the column of Hilbert
space vectors
\begin{equation}
\ket y=R^Tv_2\ket\psi.
\end{equation}
Taking the expectation of Eq.~\eqref{eq:exact-sos-main} gives
$\langle y\rvert S\lvert y\rangle=0$.  Since $S\succ0$, every component of
$\ket y$ vanishes:
\begin{equation}
R^Tv_2\ket\psi=0.
\label{eq:full-kernel-forced}
\end{equation}
Thus every degree-two relation in the ideal kernel is enforced on the unknown
maximal realization.  The physical relations used for extraction are listed
and analyzed in Appendix~\ref{app:equality-swap}.

The relation space also admits a permutation decomposition.  Its
$S_3$ multiplicities are twelve trivial, three sign, and sixteen standard
representations,
\begin{equation}
12\mathbf1\oplus3\mathbf{sgn}\oplus16\mathbf{std}.
\end{equation}
This decomposition can reduce the presentation of invariant Gram matrices
and may be useful for related symmetric constructions.

\section{From maximal violation to the SWAP isometry}
\label{app:equality-swap}

An independent program reconstructs the 55-word basis and checks that fifteen
physical relations lie in the kernel forced by
Eq.~\eqref{eq:full-kernel-forced}.  For every party $p$, they include the
linear $D$--$X$--$Z$ relation, the ordered $DX$ and $DZ$ relations,
anticommutation, and the hypergraph stabilizer.  We derive the extraction
using state-dependent operations.

Suppress the party label.  The three relations
\begin{align}
(\sqrt2D-X-Z)\ket\psi&=0,\notag\\
(\sqrt2DX-\id-ZX)\ket\psi&=0,\notag\\
(\sqrt2DZ-XZ-\id)\ket\psi&=0
\label{eq:three-calibration-relations-app}
\end{align}
are separately forced by the ideal kernel.  Left-multiply the first zero
vector by $D$ and use $D^2=\id$ to obtain
\begin{equation}
\sqrt2\ket\psi=(DX+DZ)\ket\psi.
\end{equation}
Substitution of the two ordered degree-two relations gives
\begin{equation}
2\ket\psi=(2\id+ZX+XZ)\ket\psi,
\end{equation}
and therefore $\{X,Z\}\ket\psi=0$.  At no point is the first line of
Eq.~\eqref{eq:three-calibration-relations-app} right-multiplied by $X$ or
$Z$; doing so would incorrectly promote a state relation to an operator
identity.

For distinct $\{p,q,r\}=\{A,B,C\}$, the stabilizer relation is
\begin{equation}
X_p\ket\psi=C_{qr}\ket\psi,\qquad
C_{qr}=\frac{\id+Z_q+Z_r-Z_qZ_r}{2}.
\label{eq:bell-stabilizer-app}
\end{equation}
Let $Q_{bc}=P_b^qP_c^r$.  Because $Q_{bc}$ belongs to the other parties, it
commutes with local operators of party $p$, and
\begin{equation}
C_{qr}Q_{bc}=(-1)^{bc}Q_{bc}.
\end{equation}
The state-dependent anticommutator may therefore be conditioned by $Q_{bc}$.
Direct expansion of the projectors gives
\begin{equation}
XP_a\ket\varphi=P_{1-a}X\ket\varphi
\quad\text{whenever}\quad \{X,Z\}\ket\varphi=0.
\end{equation}
Applying this identity to $\ket\varphi=Q_{bc}\ket\psi$ yields
\begin{equation}
X_A\ket{\psi_{abc}}=(-1)^{bc}\ket{\psi_{1-a,b,c}},
\end{equation}
and its cyclic variants.  This use of cross-party conditioning avoids the
generally invalid implication
$\{X_A,Z_A\}\ket\psi=0\Rightarrow
\{X_A,Z_A\}P_a^A\ket\psi=0$.

Each branch flip is implemented by a unitary reflection, so adjacent branch
norms are equal.  The branch cube is connected and the eight squared norms
sum to one; hence every branch has squared norm $1/8$.  Expanding the isometry
and applying the flips gives
\begin{equation}
X_A^aX_B^bX_C^c\ket{\psi_{abc}}
=(-1)^{abc}\ket{\psi_{000}}.
\end{equation}
With $\ket{\mathrm{junk}}=\sqrt8\ket{\psi_{000}}$, this proves state
extraction.

For measurement extraction, $Z_pP_a^p=(-1)^aP_a^p$ gives
$\Phi_pZ_p=(\sigma_z\otimes\id)\Phi_p$ identically.  State-dependent
anticommutation gives
\begin{equation}
P_0^pX_p\ket\psi=X_pP_1^p\ket\psi,\qquad
P_1^pX_p\ket\psi=X_pP_0^p\ket\psi,
\end{equation}
and hence the extracted $X$ action.  Finally, applying the first line of
Eq.~\eqref{eq:three-calibration-relations-app} to $\ket\psi$ gives
\begin{align}
\Phi D_p\ket\psi
&=\frac1{\sqrt2}\left(\Phi X_p\ket\psi+\Phi Z_p\ket\psi\right)\notag\\
&=\frac{\sigma_x^{(p)}+\sigma_z^{(p)}}{\sqrt2}
\Hccz\otimes\ket{\mathrm{junk}}.
\end{align}
This completes the analytic equality-to-isometry part of
Theorem~\ref{thm:bell-ccz}.

\section{Independent verification of the SOS identity}
\label{app:reproducibility}

The computer-assisted part of Theorem~\ref{thm:bell-ccz} is supported by a
finite machine-readable certificate and two independent exact verification
calculations.  Each calculation reconstructs the noncommutative quotient
algebra, expands the SOS identity, and compares all 784 canonical-word
coefficients using exact arithmetic in $\mathbb Q(\sqrt2)$.  Positive
definiteness of the compressed Gram matrix is proved from a rational congruence
and a rigorous Gershgorin bound.

Separate checks enumerate all 512 deterministic local assignments, evaluate
the target state and measurements, confirm the extraction relations in the
degree-two target space, and verify the branch and SWAP identities used after
maximal violation.  The arbitrary-dimensional proofs are given in
Appendices~A and~B; symbolic calculations verify the finite-dimensional
identities that occur within them.  The
certificate, calculation scripts, dependency information, and complete
reproduction instructions are provided through the repository cited in the
Code and data availability statement.

\FloatBarrier
\bibliographystyle{quantum}
\bibliography{main}

\end{document}